\documentclass[10pt,letterpaper]{article}

\usepackage{cogsci}
\usepackage{pslatex}
\usepackage{apacite}
\usepackage {amsmath,amssymb,latexsym}
\newtheorem{definition}{Definition}
\newtheorem{theorem}{Theorem}
\newenvironment{proof}{  \noindent\textit{Proof}.}{} % on peut remplacer {} par {\qed} pour forcer un carre
\def\squareforqed{\hbox{\rlap{$\sqcap$}$\sqcup$}}
\def\qed{\ifmmode\squareforqed\else{\unskip\nobreak\hfil
\penalty50\hskip1em\null\nobreak\hfil\squareforqed
\parfillskip=0pt\finalhyphendemerits=0\endgraf}\fi}

\title{Is Consciousness Computable? Quantifying Integrated Information Using Algorithmic Information Theory}
 
\author{{\large \bf Phil Maguire (pmaguire@cs.nuim.ie)} \\
  {\large \bf Philippe Moser (pmoser@cs.nuim.ie)}\\
  Department of Computer Science \\
  NUI Maynooth, Ireland
\AND {\large \bf Rebecca Maguire (rebecca.maguire@ncirl.ie)} \\
  School of Business, National College of Ireland \\
  IFSC, Dublin 1, Ireland
\AND {\large \bf Virgil Griffith (virgil@caltech.edu)} \\
  Computation and Neural Systems, \\
  Caltech, Pasadena, California}

\begin{document}

\maketitle

\begin{abstract}
In this article we review Tononi's (2008) theory of consciousness as integrated information. We argue that previous formalizations of integrated information (e.g. Griffith, 2014) depend on information loss. Since lossy integration would necessitate continuous damage to existing memories, we propose it is more natural to frame consciousness as a lossless integrative process and provide a formalization of this idea using algorithmic information theory. We prove that complete lossless integration requires noncomputable functions. This result implies that if unitary consciousness exists, it cannot be modelled computationally.   

\end{abstract}

\section{Introduction}

Continuing advances in neuroscience are allowing precise neural correlates of different aspects of consciousness to be uncovered. For example, damage to certain areas of the cortex has been shown to impair the experience of color, while other lesions can interfere with the perception of shape (Tononi, 2008). The hard question that remains is understanding how these neural correlates combine to give rise to subjective experiences.

Tononi's (2008) integrated information theory provides a theoretical framework which allows this issue to be meaningfully addressed. The theory proposes that consciousness is an information processing phenomenon and can thus be quantified in terms of a systems' organizational structure, specifically its capacity to integrate information. According to Tononi, what we mean when we say that the human brain produces consciousness is that it integrates information, thus producing behaviour which reflects the actions of a unified, singular system.

Tononi (2008) explains the foundations of his theory through two thought experiments, which we adapt slightly here. The first thought experiment establishes the requirement for a conscious observation to generate information. The second establishes the requirement for a conscious observation to be integrated with previous memories, hence generating integrated information. 

\subsection{Requirement 1: Generating Information}

Let's imagine that a factory producing scented candles invests in an artificial smell detector. The detector is used for sampling the aroma of the candles passing on the conveyor belt below and directing them to the appropriate boxes. Let's suppose that the factory is currently producing two flavors of scented candle: chocolate and lavender. In this case the detector only needs to distinguish between two possible smells. 

A batch of chocolate scented candles is passed underneath and the sensor flashes \emph{chocolate}. Can we say that the detector has actually experienced the smell of chocolate? Clearly it has managed to distinguish chocolate from lavender, but this does not guarantee that it has experienced the full aroma in the same manner as humans do. For example, it may be the case that the detector is latching onto a single molecule that separates the two scents, ignoring all other aspects. The distinction between chocolate and lavender is a binary one, and can thus be encoded by a single bit. In contrast, humans can distinguish more than 10,000 different smells detected by specialized olfactory receptor neurons lining the nose (Alberts et al., 2008). When a human identifies a smell as \emph{chocolate} they are generating a response which distinguishes between 10,000 possible states, yielding $\log_2 10,000 = 13.3$ bits of information.

The important point that Tononi (2008) raises with his initial thought experiment is that the quality of an experience is necessarily expressed relative to a range of alternative possibilities. For example, if the whole world was coloured the same shade of red, the act of labeling an object as `red' would hold no meaning. The informativeness of `red' depends on its contrast with other colours. Descriptions of experiences must be situated within a context where they discriminate among many alternatives (i.e. they must generate information). 

\subsection{Requirement 2: Generating Integrated Information}

Tononi's (2008) second thought experiment establishes that information alone is not sufficient for conscious experience. 

Imagine that the scented candle factory enhances the artificial smell detector so that now it can distinguish between 1 million different smells, even more than the human nose. Can we now say that the detector is truly smelling chocolate when it outputs \emph{chocolate}, given that it is producing more information than a human? What is the difference between the detector's experience and the human experience?

Like the human nose, the artificial smell detector uses specialized olfactory receptors to diagnose the signature of the scent and then looks it up in a database to identify the appropriate response. However, each smell is responded to in isolation of every other. The exact same response to a chocolate scent occurs even if the other 999,999 entries in the database are deleted. The factory might as well have purchased a million independent smell detectors and placed them together in the same room, each unit independently recording and responding to its own data. 

According to Tononi (2008), the information generated by such a system differs from that generated by a human insofar as it is not integrated. Because it may as well be composed of individual units, each with the most limited repertoire, an unintegrated set of responses cannot yield a subjective experience. To bind the repertoire, a system must generated integrated information. Somehow, the response to the smell of chocolate must be encoded in terms of its relationship with other experiences. 

\subsection{Consciousness as Integrated Information}

Inside the human nose there are different neurons which are specialized to respond to particular smells. The process of detection is not itself integrated. For example, with selective damage to certain olfactory receptors a person could conceivably lose their ability to smell chocolate while retaining their ability to smell lavender. However, the human experience of smell is integrated as regards the type of information it records in response.

According to Tononi's (2008) theory, when somebody smells chocolate the effect that it has on their brain is integrated across many aspects of their memory. Let's consider, for example, a human observer named Amy who has just experienced the smell of chocolate. A neurosurgeon would find it very difficult to operate on Amy's brain and eliminate this recent memory without affecting anything else. According to the integrated information theory, the changes caused by her olfactory experience are not localised to any one part of her brain, but are instead widely dispersed and inextricably intertwined with all the rest of her memories, making them difficult to reverse. This unique integration of a stimulus with existing memories is what gives experiences their subjective (i.e. observer specific) flavour. This is integrated information.

In contrast, deleting the same experience in the case of an artificial smell detector would be easy. Somewhere inside the system is a database with discrete variables used to maintain the detection history. These variables can simply be edited to erase a particular memory. The information generated by the artificial smell detector is not integrated. It does not influence the subsequent information that is generated. It lies isolated, detached and dormant. 

The same reasoning can be used to explain why a video camera, which generates plenty of information, remains unconscious, in contrast to a person viewing the same scene. The memories generated by the video camera can be easily edited independently of each other. For example, I can decide to delete all of the footage recorded yesterday between 2pm and 4pm. In contrast, a person viewing the same scenes encodes information in an integrated fashion. I cannot delete Amy's memories from yesterday because all of her memories from today have already been influenced by them. The two sets of memories cannot easily be disentangled. When it comes to human consciousness it is not possible to identify any simple divisions or disjoint components. 

What Tononi's (2008) theory proposes is that when people use the term `consciousness' to describe the behaviour of an entity they have the notion of integrated information in mind. We attribute the property of being conscious to systems whose responses cannot easily be decomposed or disintegrated into a set of causally independent parts.  In contrast, when we say that a video camera is unconscious, what we mean is that the manner in which it responds to visual stimuli is unaffected by the information it has previously recorded.

\subsection{Quantifying Integrated Information}

Tononi (2008) seeks to formalize the measurement of integrated information. His central idea is to quantify the information generated by the system as a whole above and beyond the information generated independently by its parts. For integrated information to be high, a system must be connected in such a way that information is generated by causal interactions among rather than within its parts. 

Assuming that the brain generates high levels of integrated information, this implies that the encoding of a stimulus must be deeply connected with other existing information in the brain. We now address the question of what form of processing might enable such integrated information to arise. 

Griffith (2014) rebrands the informational difference between a whole and the union of its parts as `synergy'. He presents the XOR gate as the canonical example of synergistic (i.e. integrated) information. Consider, for example, a XOR gate with two inputs, $X_1$ and $X_2$, which can be interpreted as representing a stimulus and an original brain state. They combine integratively to yield $Y$, the resultant brain state which encodes the stimulus. Given $X_1$ and $X_2$ in isolation we have no information about $Y$. The resultant brain state $Y$ can only be predicted when both components are taken into account at the same time. Given that the components $X_1$ and $X_2$ do not have any independent causal influence on $Y$, all of the information about $Y$ here is integrated. 

%\begin{table}[!ht]
%\begin{center} 
%\caption{XOR gate inputs $X_1$ and $X_2$ provide synergistic information on output $Y$} 
%\label{sample-table} 
%\vskip 0.12in
%\begin{tabular}{lll} 
%\hline
%$X_1$    & $X_2$ & $Y$ \\
%\hline
%0       &   0 & 0\\
%0   &   1  & 1\\
%1        &   0 & 1\\
%1          &   1 & 0\\
%\hline
%\end{tabular} 
%\end{center} 
%\end{table}

One issue with presenting the XOR gate as the canonical example of synergistic information is that it is lossy. A two bit input is reduced to a single bit output, meaning that half the entropy has been irretrievably lost. If the brain integrated information in this manner, the inevitable cost would be the destruction of existing information. While it seems intuitive for the brain to discard irrelevant details from sensory input, it seems undesirable for it to also hemorrhage meaningful content. In particular, memory functions must be vastly non-lossy, otherwise retrieving them repeatedly would cause them to gradually decay. 

We propose that the information integration evident in cognition is not lossy. In the following sections we define a form of synergy, based on data compression, which does not rely on the destruction of information, and subsequently explore its implications.

\section{Data Compression as Integration}

Data compression is the process by which an observation is reduced by identifying patterns within it. For example the sequence ${4, 6, 8, 12, 14, 18, 20, 24\ldots}$ can be simplified as the description ``odd prime numbers +1''. The latter representation is shorter than the original sequence, hence it evidences data compression. 

A close link exists between data compression and prediction. Levin's (1974) Coding Theorem demonstrates that, with high probability, the most likely model that explains a set of observations is the most compressed one. In addition, for any predictable sequence of data, the optimal prediction of the next item converges quickly with the prediction made by the model which has the simplest description (Solomonoff, 1964). As per Occam's razor, concise models make fewer assumptions and are thus more likely to be correct. 

These insights lay the foundation for a deep connection between data compression, prediction and understanding, a theoretical perspective on intelligence and cognisance which we refer to as `compressionism'. Adopting this perspective, Maguire and Maguire (2010) propose that the binding of information we associate with consciousness is achieved through sophisticated data compression carried out in the brain, suggesting a link between this form of processing and Tononi's (2008) notion of information integration.

In the case of an uncompressed string, every bit carries independent information about the string. In contrast, when a text file is compressed to the limit, each bit in the final representation is fully dependent on every other bit for its significance. No bit carries independent information about the original text file. For an uncompressed file, damaging the first bit leaves you with a 50\% chance of getting the first bit right and 100\% chance of getting the rest of the bits right. For an optimally compressed file, damaging the first bit corrupts everything and leaves you with only a 50\% chance of getting all the bits right and a 50\% chance of getting them all wrong. Clearly, the information encoded by the bits in the compressed file is more than the sum of its parts, highlighting a link between data compression and Tononi's (2008) concept of integrated information.

In the following section we formally prove that, given the Partial Information Decomposition (Williams \& Beer, 2010) formulation of synergy, the amount of integrated information an information-lossless process produces on statistically independent inputs is equivalent to the data compression it achieves. 

We begin with a brief description of algorithmic information theory (see Li and Vity\'{a}nyi, 2008, for more details). We use \emph{strings} to refer to finite binary sequence, i.e. an element of set
$2^{< \omega}$. Any finite object can be encoded into a string in some natural way.
We are interested in effective 
descriptions of strings (i.e. computable by a universal computer i.e. Turing machine) . For a string $x$, its (plain) Kolmogorov complexity $C(x)$ is the length of the shortest effective description
of $x$. More formally, fix a universal Turing machine $U$. $C(x)$ is the length of the shortest 
program $x^*$ such that $U$ on input $x^*$ outputs $x$. It can be shown that the value of $C(x)$ does not depend on the choice of 
$U$ up to an additive constant.  $C(x)$ is the amount of algorithmic information contained in $x$.
A random string is a string $x$ that cannot be compressed, e.g. such that $C(x)$ is at least the length of $x$.

For two strings $x,y$ the conditional Kolmogorov complexity $C(x|y)$ of $x$ given $y$
is the size of the shortest program $q$ such that $U$ on input $p$ and provided $y$ as an extra input, outputs $x$.
The information $x$ has about $y$ is defined as 
$I(x:y) = C(x) - C(x|y) =^+ C(y) - C(y|x)$, where $=^+$ means equal up to a $O(1)$ (constant) term. The idea of C-based synergy (Griffith, 2014) 
is to define four intuitive slices of the C-information of the function $m: (x_1,x_2) \mapsto y$.
\begin{enumerate}
    \item R: the amount of the C-information strings $x_1$ and $x_2$ convey redundantly about $y$, or, equivalently, the amount of data compression that $y$ achieves assuming statistically independent inputs
    \item $U_1$: the amount of C-information that only string $x_1$ conveys about $y$.
    \item $U_2$: the amount of C-information that only string $x_2$ conveys about $y$.
    \item $S$: the amount of C-information the concatenation string, $(x_1,x_2)$ conveys about $y$ not conveyed by either $x_1$ or $x_2$.
\end{enumerate}

From the Partial Information Decomposition framework (Williams \& Beer, 2010), we have the following equalities relating the nonnegative scalars $R$, $U_1$, $U_2$, and $S$:
\begin{equation*}
\begin{split}
    R + U_1 &= I(x_1:y) \\
    R + U_2 &= I(x_2:y) \\
    R + U_1 + U_2 + S &= I(x_1,x_2:y) \; .
\end{split}
\end{equation*}

First, using the three equalities above we can define an easy expression for the synergy minus the redundancy,
\begin{align*}
    I(x_1,x_2:y) - I(x_1:y) - I(x_2:y)  &= S - R \; .
\end{align*}

\begin{theorem}
 Given  $C(x_1,x_2|y)=0$, then $S \leq R$ with equality when $I(x_1:x_2) = 0$.
\end{theorem}
\begin{proof} 
Using the prior expression we expand the three C-information slices into their respective C-entropies.
\begin{align*}
    S - R &= I(x_1,x_2:y) - I(x_1:y) - I(x_2:y) \\
    &= C(x_1,x_2) - C(x_1) - C(x_2) - C(x_1,x_2|y) + C(x_1|y) + \\
	&C(x_2|y) .
\end{align*}
Given that $C(x_1,x_2|y)=0$, we know likewise that $C(x_1|y)=C(x_2|y) = 0$; we simplify the above,
\begin{align*}
    S - R &=  C(x_1,x_2) - C(x_1) - C(x_2) \\
    &=  C(x_1) + C(x_2|x_1) - C(x_1) - C(x_2) \\
   &= -I(x_1:x_2).
\end{align*}
From the above we have,
\[
    S = R - I(x_1:x_2) \; .
\]
Which entails $S \leq R$ with equality when $I(x_1:x_2) = 0$.
\qed
\end{proof}

~

The above result shows that synergy (i.e. integrated information) is equivalent to redundancy (i.e. data compression) for lossless functions operating on statistically independent inputs. However, an obstacle remains to expressing synergy in this format. Although Griffith's (2014) formulation of synergy identifies the link with data compression, giving a definition of the C-information slices $R,U_1,U_2,S$ based
on $C$-complexity is not trivial.

To quantify synergy for lossless functions using C-complexity, Tononi's (2008) definition of integrated infromation must be somehow translated from its original operational framework of Shannon information theory to that of algorithmic information theory. We now show that the most natural way of performing this translation does not succeed.

Suppose the synergy of function $(x,y) \mapsto z$ is defined as
$$S_0(x,y:z) = C(z|x) + C(z|y) - C(z|xy) - C(z|x\cap y)$$
where $C(z|x\cap y)$ is the shortest program that outputs $z$ given advice $x$ or $y$, (i.e. the program
outputs $z$ on any of the two advices $x$ or $y$). 
The following result shows that, using this definition, the concatenation function turns out to have high $S_0$ synergy, which is anomalous.
\begin{theorem}
Consider the concatenation function $z(x,y)=xy$. Then $z$ is a lossless function of $S_0$ synergy $|z|/2$.
\end{theorem}
\begin{proof}
Pick two independent $n/2$-bit random strings $x,y$ starting with $0$ resp. $1$
i.e. $x=0\ldots$, $y=1\ldots$ and $C(x|y)=n/2$ and $C(y|x)=n/2$.

By definition of synergy
$$S_0(x,y:z) = C(z|x) + C(z|y) - C(z|xy) - C(z|x\cap y)$$
where the first two terms are $n/2$, the third is $O(1)$, and the last is $n/2$ because
of the following program $p$. $p$ is an $O(1)$ instructions part followed by the bitwise 
XOR of $x,y$ denoted $w$, i.e. $n/2 +O(1)$ bits total.
Instructions: Given advice $a$, XOR $a$ with $w$ to obtain $d$. If $d$ starts with $0$
output $da$, else output $ad$. So when $a=x$, $d=y$ and we output $ad=xy$. Similarly when
$a=y$ then $d=x$ and we output $da=xy$, i.e. $C(z|x\cap y)=n/2$.
\qed
\end{proof}

~

In the following section, we outline an alternative strategy for defining integrated information using C-complexity.

\subsection{Quantifying Integration Using Edit Distance}

If data is optimally compressed then it becomes extremely difficult to edit in its compressed state. For example, imagine a compressed encoding of a Wikipedia page. You want to edit the first word on the page. But where is this word encoded in the compressed file? There is no easily delineated set of bits which corresponds to the first word and nothing else. Instead, the whole set of data has been integrated, with every bit from the original file depending on all the others. To discern the impact that the first word has had on the compressed encoding you have to understand the compression. There are no shortcuts. 

To formalize integrated information as data compression we consider a stimulus, first in its raw unintegrated state, and second, encoded in its integrated state within the brain. The level of integration is equivalent to the difficulty of identifying the raw information and editing it within its integrated state. 

In the following definition $z$ and $\bar z$ are the raw stimulus and the brain encoded stimulus. We consider the difficulty of editing $z$ into $z'$, for example, editing the smell of chocolate to turn it into the smell of lavender. If this operation is performed on a raw, unintegrated dataset then the task is straight-forward: the bits that differ are simply altered. Consider, however, the challenge for the neurosurgeon operating on Amy's brain. If the stimulus has not been widely integrated then the neurosurgeon can concentrate on a single localised area of her brain and hopefully the encoding will be overt, reflecting the original unintegrated format in which the information was originally transmitted. However, if the stimulus has been successfully integrated (i.e. compressed) then its encoding will reflect the overlap of patterns between it and the entire contents of Amy's brain. Its representation will be widely distributed, with effects on all kinds of other memories, making it impossible to isolate and edit. 

We quantify the integration of an encoding process operating on a stimulus as the minimum informational distance between the original state of the encoded stimulus and any possible edited state. If every state is completely different to the original, then the integration is 1; if there exists an edited state which is only trivially removed, the integration is 0.

For example, when an image on a digital camera is altered, the informational distance between the camera's original and edited state is small. In contrast, the neurosurgeon struggles to edit the memories in Amy's brain: changing even the slightest detail requires the contents of her brain to be completely reconstructed. The edit distance is so great that her original brain state is largely useless for identifying a target edited brain state. 

Formally, the edit distance of $m$ at point $z$ is a number between 0 and 1 that measures the level
of integration of $m(z)$. It is measured by looking at all strings $z'$ similar to $z$, and finding
the one that minimizes the ratio of length of the shortest description of $m(z)$ given $m(z')$ to the
length of shortest description of $m(z)$. The smallest ratio obtained is the edit distance.
Since the numerator is always positive and less or equal to the denominator, the edit distance
is between 0 and 1. This edit distance quantifies information integration for lossless functions.

\begin{definition}
The edit distance of $m$ at point $z$ is given by $$\min_{z'\neq z: C(z|z')\leq \log |z|} \{\frac{C(m(z)|m(z'))}{C(m(z))}\}.$$
\end{definition}

\section{On the Computability of Integration}

In this section we prove an interesting result using the above definition, namely that lossless information integration cannot be achieved by a computable process. 

According to the integrated information theory, when we think of another person as conscious we are viewing them as a completely integrated and unified information processing system, with no feasible means of disintegrating their conscious cognition into disjoint components. We assume that their behaviour calls into play all of their memories and reflects full coordination of their sensory input. We now prove that this form of complete integration cannot be modelled computationally.

An \emph{integrating} function's output is such that the information
of its two (or more) inputs is completely integrated. More formally,
\begin{definition}
A 1-1 function $m: z=(z_1,z_2) \mapsto \bar z$ is integrating if for any strings $z\neq z'$,
$C(\bar z' \ | \ \bar z) \geq C(\bar z') - C(z' \ | \ z)$.
\end{definition}
i.e, the knowledge of $m(z)$ does not help to describe $m(z')$, when $z$ and $z'$ are close.

\begin{theorem}
No integrating function is computable.
\end{theorem}
\begin{proof}
Suppose $m$ is a computable integrating function.
Let $z$ be a random string, i.e. such that $C(z)\geq |z|$.
Let $z'$ be the string obtained by flipping the first bit of $z$.
We have $C(z' \ | \ z) = O(1)$.
Consider the following program for $\bar z'$ given $\bar z$:
Cycle through all strings until the unique $z$ is found such that $m(z)=\bar z$.
Compute $z'$ by flipping the first bit of $z$.
Compute $\bar z' = m(z')$.

Since $m$ is computable, the program above is of constant size i.e.,
$C(\bar z' \ | \ \bar  z) = O(1)$.
Also $C(\bar z') =^+ C(z') =^+ C(z) \geq |z|$ because $m$ is computable, 1-1 and by choice of $z$.

Because $m$ is integrating, we have $C(\bar z' \ | \ \bar z)\geq C(\bar z') - C(z' \ | \ z) = |z| - O(1)$,
a contradiction.
\qed
\end{proof}

~

The implications of this proof are that we have to abandon either the idea that people enjoy genuinely unitary consciousness or that brain processes can be modelled computationally. 

If a person's behaviour is totally resistant to disintegration (i.e. we cannot analyse it independently from the rest of their cognition), then it implies that something is going on in their brain that is so complex it cannot feasibly be reversed. In line with this view, Bringsjord and Zenzen (1997) specifically argue that the difference between cognition and computation is that computation is reversible whereas cognition is not. For instance, it is impossible for the neurosurgeon to operate on Amy's brain and directly edit her conscious memories, because the process of integration is irreversibly complex. 

Yet Amy's brain is a physical causal system which follows the laws of physics. Information flows into Amy's brain conducted by nerve impulses and gets processed by neurons through biochemical signalling. Whatever information-lossless changes result should theoretically be reversible. To argue otherwise seems to suggest that a form of magic is going on in the brain, which is beyond computational modelling.   

McGinn (1991) points out that intractable complexity of the mind does not necessarily require the brain to transcend the laws of physics: instead, the intractability can have an observer specific source. He argues that the mind-body problem is cognitively closed to humans in the same way that quantum mechanics is closed to a zebra. This perspective, known as `new mysterianism', maintains that the hard problem of consciousness stems, not from a supernatural process, but from natural limits in how humans form concepts. 

Similarly, the apparent unitary nature of consciousness does not require a mystical process of integration which transcends physical computability. Our result merely establishes a link between integration and irreversibility, the cause of which can be due to limitations in the observer's perspective. 
While we intuitively assume that consciousness must be a fundamental property as defined from a God's eye perspective, the attribution of this property always takes place in a social context. When people attribute consciousness to a system they are acknowledging a \emph{subjective} inability to break it down into a set of independent components, forcing them to treat its actions as the behaviour of a unified, integrated whole. The irreversibilty here is observer-centric, as opposed to absolute.
Rather than establishing a new property of consciousness, our result can therefore be interpreted as merely clarifying what is meant by the use of this concept. Specifically, conscious behaviour is that which is resistant to our best attempts at decomposition.

\subsection{Neuroscientific Modelling}

An alternative account is simply that consciousness does not exist: the unitary appearance of people's behaviour is recognizable as an illusion. Dennett (1991) adopts this perspective with his multiple drafts model. He views consciousness as being inherently decomposable, criticizing the idea of what he calls the `Cartesian theatre', a point where all of the information processing in the brain is integrated. Dennett presents consciousness as a succession of multiple drafts, a process in constant flux, without central organization or irreversible binding.

Could neuroscience provide us with a mechanical model of human behaviour that supersedes the value of attributing consciousness, as Dennett (1991) suggests? 
It sometimes arises that a system to which we have previously attributed unitary consciousness is subsequently recognized as following mechanical rules. For example, when conversing with a chatbot, we might suddenly notice that its responses can be predicted solely on the basis on the preceding sentence. We then adopt the superior rule-based model and cease to attribute consciousness.

Ultimately, for consciousness to be revealed as an illusion, people would have to agree that neuroscientific modelling succeeds in disintegrating every aspect of behaviour. Note that the key word here is `people': people would have to agree. Arguably, the ultimate standard that we have for measurement depends on the notion of other observers, which are themselves integrated, unified wholes. For this reason, Maguire and Maguire (2011) speculate that future developments in information theory will recognize the intractable complexity of the mind as a key concept supporting the notion of objectivity in measurement, a shift which would undermine the meaningfulness of the goal to `understand' the mind. 

\subsection{Scramble In, Scramble Out}

Assuming integrated consciousness is a genuine phenomenon, its noncomputability has interesting implications for what has to happen in the brain. When stimuli are picked up by the brain they enter at disintegrated locations. For example, visual stimuli enter through the optic nerve and are processed initially by the primary visual cortex. When a visual stimulus is encoded in the occipital lobe it clearly has not yet been integrated with the rest of cognition. For instance, Stanely, Li and Dan (1999) analysed an array of electrodes embedded in the thalamus lateral geniculate nucleus area of a cat and were able to decode the signals to generate watchable movies of what the cat was observing. 

Similarly, the initiation of action must be localised in particular areas of the brain which control the relevant muscles. This readiness potential must detach from the rest of the brain's processing and hence is no longer integrated. For example, following up on Libet's original experiments, Siong Soon et al. (2008) demonstrated that, by monitoring activity in the frontopolar prefrontal cortex they could predict a participant's decision to move their right or left hand several seconds before the participant became aware of it.   

However, if integration is necessary for consciousness, then somewhere between the stimulus entering the brain and the decision leaving the brain, there is a point where the information cannot be fully disentangled from the rest of cognition. This integrated processing cannot be localised to any part of the brain or any specific point in time. The contents of cognition are effectively unified. We label this idea `scramble in, scramble out' to reflect the irreversible integration and disintegration that must occur between observation and action. 

The aspects of cognition that have been clarified by neuroscience so far tend to involve processing before scramble in or after scramble out. For example, it is well established that the occipital lobe is involved in visual processing or that the prefrontal cortex encodes future actions before they are performed. These components are modular in that they have specialised, encapsulated, evolutionarily developed functions. However, somewhere between input and output there must also be a binding process of integration that no computational modelling can disentangle. 

Fodor (2001) summarizes as follows: ``Local mental processes appear to accommodate pretty well to Turing's theory that thinking is computation; they appear to be largely modular...By contrast, what we've found out about global cognition is mainly that it is different from the local kind...we deeply do not understand it''. While neuroscience might shed light on the input and output functions of the brain, the quantification for integrated information we have presented here implies that it will be unable to shed light on the complex tangle that is core consciousness.

\nocite{*}

\bibliographystyle{apacite}

\setlength{\bibleftmargin}{.125in}
\setlength{\bibindent}{-\bibleftmargin}

\bibliography{References3}

\begin{thebibliography}{}

\bibitem[\protect\citeauthoryear{%
Alberts%
\ \protect\BOthers{.}}{%
Alberts%
\ \protect\BOthers{.}}{%
{\protect\APACyear{1997}}%
}]{%
alberts1997molecular}%
\APACinsertmetastar{%
alberts1997molecular}%
Alberts, B.%
, Johnson, A.%
, Lewis, J.%
, Raff, M.%
, Roberts, K.%
\BCBL{}\ \BBA{} Walter, P.%
%
\unskip\
\newblock
\APACrefYear{1997}.
\newblock
\APACrefbtitle{Molecular Biology of the Cell}{Molecular biology of the cell}.
\PrintBackRefs{\CurrentBib}

\bibitem[\protect\citeauthoryear{%
Bringsjord%
\ \BBA{} Zenzen%
}{%
Bringsjord%
\ \BBA{} Zenzen%
}{%
{\protect\APACyear{1997}}%
}]{%
bringsjord1997cognition}%
\APACinsertmetastar{%
bringsjord1997cognition}%
Bringsjord, S.%
\BCBT{}\ \BBA{} Zenzen, M.%
%
\unskip\
\newblock
\APACrefYearMonthDay{1997}{}{}.
\newblock
\BBOQ{}\APACrefatitle{Cognition is not computation: The argument from
  irreversibility}{Cognition is not computation: The argument from
  irreversibility}.\BBCQ{}
\newblock
\APACjournalVolNumPages{Synthese}{113}{2}{285--320}.
\PrintBackRefs{\CurrentBib}

\bibitem[\protect\citeauthoryear{%
Dennett%
}{%
Dennett%
}{%
{\protect\APACyear{1991}}%
}]{%
dennett1991consciousness}%
\APACinsertmetastar{%
dennett1991consciousness}%
Dennett, D\BPBI C.%
%
\unskip\
\newblock
\APACrefYear{1991}.
\newblock
\APACrefbtitle{Consciousness Explained}{Consciousness explained}.
\newblock
\APACaddressPublisher{}{Little, Brown}.
\PrintBackRefs{\CurrentBib}

\bibitem[\protect\citeauthoryear{%
Fodor%
}{%
Fodor%
}{%
{\protect\APACyear{2001}}%
}]{%
fodor2001mind}%
\APACinsertmetastar{%
fodor2001mind}%
Fodor, J\BPBI A.%
%
\unskip\
\newblock
\APACrefYear{2001}.
\newblock
\APACrefbtitle{The Mind Doesn't Work That Way: The Scope and Limits of
  Computational Psychology}{The mind doesn't work that way: The scope and
  limits of computational psychology}.
\newblock
\APACaddressPublisher{}{MIT press}.
\PrintBackRefs{\CurrentBib}

\bibitem[\protect\citeauthoryear{%
Griffith%
}{%
Griffith%
}{%
{\protect\APACyear{2014}}%
}]{%
griffith}%
\APACinsertmetastar{%
griffith}%
Griffith, V.%
%
\unskip\
\newblock
\APACrefYearMonthDay{2014}{}{}.
\newblock
\BBOQ{}\APACrefatitle{A principled infotheoretic $\Psi$-like measure}{A
  principled infotheoretic $\psi$-like measure}.\BBCQ{}
\newblock
\APACjournalVolNumPages{arXiv preprint arXiv:1401.0978}{}{}{}.
\PrintBackRefs{\CurrentBib}

\bibitem[\protect\citeauthoryear{%
Legg%
\ \BBA{} Hutter%
}{%
Legg%
\ \BBA{} Hutter%
}{%
{\protect\APACyear{2007}}%
}]{%
legg2007tests}%
\APACinsertmetastar{%
legg2007tests}%
Legg, S.%
\BCBT{}\ \BBA{} Hutter, M.%
%
\unskip\
\newblock
\APACrefYearMonthDay{2007}{}{}.
\newblock
\BBOQ{}\APACrefatitle{Tests of machine intelligence}{Tests of machine
  intelligence}.\BBCQ{}
\newblock
\BIn{} \APACrefbtitle{50 Years of Artificial Intelligence}{50 years of
  artificial intelligence}\ (\BPGS\ 232--242).
\newblock
\APACaddressPublisher{}{Springer}.
\PrintBackRefs{\CurrentBib}

\bibitem[\protect\citeauthoryear{%
Levin%
}{%
Levin%
}{%
{\protect\APACyear{1974}}%
}]{%
Levin1974}%
\APACinsertmetastar{%
Levin1974}%
Levin, L\BPBI A.%
%
\unskip\
\newblock
\APACrefYearMonthDay{1974}{}{}.
\newblock
\BBOQ{}\APACrefatitle{Laws of information conservation (non-growth) and aspects
  of the foundation of probability theory}{Laws of information conservation
  (non-growth) and aspects of the foundation of probability theory}.\BBCQ{}
\newblock
\APACjournalVolNumPages{Problems Information Transmission}{10}{3}{206--210}.
\PrintBackRefs{\CurrentBib}

\bibitem[\protect\citeauthoryear{%
Li%
\ \BBA{} Vit{\'a}nyi%
}{%
Li%
\ \BBA{} Vit{\'a}nyi%
}{%
{\protect\APACyear{2008}}%
}]{%
li2008introduction}%
\APACinsertmetastar{%
li2008introduction}%
Li, M.%
\BCBT{}\ \BBA{} Vit{\'a}nyi, P.%
%
\unskip\
\newblock
\APACrefYear{2008}.
\newblock
\APACrefbtitle{An Introduction to Kolmogorov Complexity and its
  Applications}{An introduction to kolmogorov complexity and its applications}.
\newblock
\APACaddressPublisher{}{Springer}.
\PrintBackRefs{\CurrentBib}

\bibitem[\protect\citeauthoryear{%
Maguire%
\ \BBA{} Maguire%
}{%
Maguire%
\ \BBA{} Maguire%
}{%
{\protect\APACyear{2010}}%
}]{%
maguireconsciousness}%
\APACinsertmetastar{%
maguireconsciousness}%
Maguire, P.%
\BCBT{}\ \BBA{} Maguire, R.%
%
\unskip\
\newblock
\APACrefYearMonthDay{2010}{}{}.
\newblock
\BBOQ{}\APACrefatitle{Consciousness is data compression}{Consciousness is data
  compression}.\BBCQ{}
\newblock
\BIn{} \APACrefbtitle{Proceedings of the Thirty-Second Conference of the
  Cognitive Science Society}{Proceedings of the thirty-second conference of the
  cognitive science society}\ (\BPGS\ 748--753).
\PrintBackRefs{\CurrentBib}

\bibitem[\protect\citeauthoryear{%
Maguire%
\ \BBA{} Maguire%
}{%
Maguire%
\ \BBA{} Maguire%
}{%
{\protect\APACyear{2011}}%
}]{%
maguire2011understanding}%
\APACinsertmetastar{%
maguire2011understanding}%
Maguire, P.%
\BCBT{}\ \BBA{} Maguire, R.%
%
\unskip\
\newblock
\APACrefYearMonthDay{2011}{}{}.
\newblock
\BBOQ{}\APACrefatitle{Understanding the Complexity of the Mind}{Understanding
  the complexity of the mind}.\BBCQ{}
\newblock
\APACjournalVolNumPages{European Perspectives on Cognitive Science}{}{}{}.
\PrintBackRefs{\CurrentBib}

\bibitem[\protect\citeauthoryear{%
McGinn%
}{%
McGinn%
}{%
{\protect\APACyear{1991}}%
}]{%
mcginn1991problem}%
\APACinsertmetastar{%
mcginn1991problem}%
McGinn, C.%
%
\unskip\
\newblock
\APACrefYear{1991}.
\newblock
\APACrefbtitle{The Problem of Consciousness: Essays Towards a Resolution}{The
  problem of consciousness: Essays towards a resolution}.
\newblock
\APACaddressPublisher{}{Blackwell Oxford, UK}.
\PrintBackRefs{\CurrentBib}

\bibitem[\protect\citeauthoryear{%
Solomonoff%
}{%
Solomonoff%
}{%
{\protect\APACyear{1964}}%
}]{%
solomonoff1964formal}%
\APACinsertmetastar{%
solomonoff1964formal}%
Solomonoff, R\BPBI J.%
%
\unskip\
\newblock
\APACrefYearMonthDay{1964}{}{}.
\newblock
\BBOQ{}\APACrefatitle{A formal theory of inductive inference. Part I}{A formal
  theory of inductive inference. part i}.\BBCQ{}
\newblock
\APACjournalVolNumPages{Information and Control}{7}{1}{1--22}.
\PrintBackRefs{\CurrentBib}

\bibitem[\protect\citeauthoryear{%
Soon%
, Brass%
, Heinze%
\BCBL{}\ \BBA{} Haynes%
}{%
Soon%
\ \protect\BOthers{.}}{%
{\protect\APACyear{2008}}%
}]{%
soon2008unconscious}%
\APACinsertmetastar{%
soon2008unconscious}%
Soon, C\BPBI S.%
, Brass, M.%
, Heinze, H\BHBI J.%
\BCBL{}\ \BBA{} Haynes, J\BHBI D.%
%
\unskip\
\newblock
\APACrefYearMonthDay{2008}{}{}.
\newblock
\BBOQ{}\APACrefatitle{Unconscious determinants of free decisions in the human
  brain}{Unconscious determinants of free decisions in the human brain}.\BBCQ{}
\newblock
\APACjournalVolNumPages{Nature Neuroscience}{11}{5}{543--545}.
\PrintBackRefs{\CurrentBib}

\bibitem[\protect\citeauthoryear{%
Stanley%
, Li%
\BCBL{}\ \BBA{} Dan%
}{%
Stanley%
\ \protect\BOthers{.}}{%
{\protect\APACyear{1999}}%
}]{%
stanley1999reconstruction}%
\APACinsertmetastar{%
stanley1999reconstruction}%
Stanley, G\BPBI B.%
, Li, F\BPBI F.%
\BCBL{}\ \BBA{} Dan, Y.%
%
\unskip\
\newblock
\APACrefYearMonthDay{1999}{}{}.
\newblock
\BBOQ{}\APACrefatitle{Reconstruction of natural scenes from ensemble responses
  in the lateral geniculate nucleus}{Reconstruction of natural scenes from
  ensemble responses in the lateral geniculate nucleus}.\BBCQ{}
\newblock
\APACjournalVolNumPages{The Journal of Neuroscience}{19}{18}{8036--8042}.
\PrintBackRefs{\CurrentBib}

\bibitem[\protect\citeauthoryear{%
Tononi%
}{%
Tononi%
}{%
{\protect\APACyear{2008}}%
}]{%
tononi2008consciousness}%
\APACinsertmetastar{%
tononi2008consciousness}%
Tononi, G.%
%
\unskip\
\newblock
\APACrefYearMonthDay{2008}{}{}.
\newblock
\BBOQ{}\APACrefatitle{Consciousness as integrated information: a provisional
  manifesto}{Consciousness as integrated information: a provisional
  manifesto}.\BBCQ{}
\newblock
\APACjournalVolNumPages{The Biological Bulletin}{215}{3}{216--242}.
\PrintBackRefs{\CurrentBib}

\bibitem[\protect\citeauthoryear{%
Williams%
\ \BBA{} Beer%
}{%
Williams%
\ \BBA{} Beer%
}{%
{\protect\APACyear{2010}}%
}]{%
williams2010nonnegative}%
\APACinsertmetastar{%
williams2010nonnegative}%
Williams, P\BPBI L.%
\BCBT{}\ \BBA{} Beer, R\BPBI D.%
%
\unskip\
\newblock
\APACrefYearMonthDay{2010}{}{}.
\newblock
\BBOQ{}\APACrefatitle{Nonnegative decomposition of multivariate
  information}{Nonnegative decomposition of multivariate information}.\BBCQ{}
\newblock
\APACjournalVolNumPages{arXiv preprint arXiv:1004.2515}{}{}{}.
\PrintBackRefs{\CurrentBib}

\end{thebibliography}

\end{document}